\documentclass[11pt, onecolumn]{IEEEtran}
\usepackage{amsmath}
\usepackage{amssymb}
\usepackage{verbatim}
\usepackage{epsfig}
\usepackage{subfigure}

\psfull

\newtheorem{theorem}{Theorem}

\newtheorem{remark}{\indent \bf Remark}[section]

\newcommand{\Eb}{\frac{E_b}{N_0}}
\newcommand{\Ebp}{\left( \Eb \right)}
\newcommand{\Ebt}{\tilde{\mathcal{E}}}
\newcommand{\Ce}{{\sf C}}

\def\SINR {\mathrm{SINR}}
\def\ex {\mathrm{e}}
\begin{document}

\title{Bandwidth Partitioning in Decentralized Wireless Networks}
\author{Nihar Jindal,~\IEEEmembership{Member,~IEEE,} Jeffrey G. Andrews,~\IEEEmembership{Senior Member,~IEEE,} \\ Steven Weber,~\IEEEmembership{Member,~IEEE}
\thanks{The authors are with the ECE Departments of the University of
Minnesota, the University of Texas at Austin, and Drexel
University, respectively. e-mail: nihar@umn.edu,
jandrews@ece.utexas.edu, sweber@ece.drexel.edu.}
\thanks{This work is supported under an NSF collaborative research 
grant awarded to the three authors (NSF grant \#0635003 (Weber), \#0634979
(Andrews), and \#0634763 (Jindal)), and by the DARPA IT-MANET
program, Grant W911NF-07-1-0028 (Andrews, Jindal, Weber). Earlier
versions of this work appeared at the ITA workshop \cite{JinAndWeb_ITA07}
and ISIT \cite{JinAndWeb_ISIT07} in Jan. and June 2007, respectively.}
}
% make the title area
\maketitle

\begin{abstract}
This paper addresses the following question, which is of interest in
the design of a multiuser decentralized network. Given a total
system bandwidth of $W$ Hz and a fixed data rate constraint of $R$
bps for each transmission, how many frequency slots $N$ of size
$W/N$ should the band be partitioned into in order to maximize the
number of simultaneous links in the network?  Dividing the available
spectrum results in two competing effects.  On the positive side, a
larger $N$ allows for more parallel, non-interfering communications
to take place in the same area. On the negative side, a larger $N$
increases the SINR requirement for each link because the same
information rate must be achieved over less bandwidth. Exploring
this tradeoff and determining the optimum value of $N$ in terms of
the system parameters is the focus of the paper. Using stochastic
geometry, the optimal SINR threshold -- which directly corresponds
to the optimal spectral efficiency -- is derived for both the low
SNR (power-limited) and high SNR (interference-limited) regimes.
This leads to the optimum choice of the number of frequency bands
$N$ in terms of the path loss exponent, power and noise spectral
density, desired rate, and total bandwidth.

\end{abstract}

\section{Introduction}
For purposes of wireless communication, the electromagnetic
spectrum is typically first divided into a large number of bands
by regulatory agencies such as the FCC or the European Commission.
These bands are typically allocated by executive fiat or auction,
and for particular purposes.
Once allocated, these bands are usually further divided into many
smaller bands that individual users have access to.  This entire
process has a major impact on the efficiency with which spectral
resources are used, but historically appears to have been done in a
mostly ad hoc manner. This paper attempts to develop a theoretical
basis for bandwidth partitioning, in particular the second
partitioning of an allocated band into subbands.

To be more specific, consider a spatially distributed wireless
network, representing either an ad hoc network or an unlicensed (and
uncoordinated) spectrum system, e.g., 802.11. In such systems it is
common to have a fixed total bandwidth, a large number of potential
users, and a limit on acceptable packet loss rates. It is also
typical to have a target data rate for each user, either to support
a certain application or due to user expectations.
This gives rise to the following basic question: given bandwidth $W$
and a fixed rate requirement $R$ for each transmitter-receiver link
in the network, how many slots $N$ of size $W/N$ should this band be
partitioned into in order to maximize the number of links (i.e.,
spatial density of transmissions) that can achieve this rate $R$ at
a specified outage probability (i.e., packet error rate)?

For example, given 1 MHz of bandwidth and a desired rate of 1 Mbps,
should  (a) each transmitter utilize the entire spectrum and thus
require an SINR of 1 (utilizing $R=W \log_2(1 + \SINR)$ bits/sec),
(b) the band be split into two 0.5 MHz sub-bands where each
transmitter utilizes one of the sub-bands with a required SINR of 3,
or (c) the band be split into $N > 2$ orthogonal $\frac{1}{N}$ MHz
sub-bands where each transmitter utilizes one of the sub-bands with
a required SINR of $2^N-1$?

Increasing the number of sub-bands $N$ has two competing effects. On
the positive side, it allows for parallel, non-interfering
communications on different sub-bands.  On the negative side,
transmitting at the same data rate over less bandwidth requires each
transmission to be performed at a higher spectral efficiency ($R$
bps over $\frac{W}{N}$ Hz corresponds to a spectral efficiency of
$\frac{R}{W/N}$ bps/Hz), which translates to a higher SINR
requirement and thus a larger interference-free area.
The objective of this paper is understanding this tradeoff and
characterizing the optimum value of $N$ in terms of the system
parameters.

%%%%%%%%%%%%%%%%%%%%%%%%%%%%%%%%%%%%%%%%%%%%%%%%
\subsection{Technical Approach}
To allow for analytical tractability, we optimize the number of
sub-bands for a network consisting of transmitter-receiver pairs
distributed on the two-dimensional plane. More specifically, the
network we consider has the following key characteristics:
\begin{itemize}
    \item Transmitter locations are a realization of a homogeneous
    spatial Poisson process.
    \item Each transmitter communicates with a single receiver
that is a distance $d$ meters away.
    \item All transmissions occur at power $\rho$ and rate $R$ bits/sec, the
    noise spectral density is $N_0$, and attenuation follows path-loss
    exponent $\alpha$.
    \item The system bandwidth of $W$ Hz is divided into $N$ equal sub-bands of $\frac{W}{N}$
    Hz, and each transmission occurs on a randomly chosen sub-band.
    \item Each receiver treats multi-user interference as noise, and
    thus a transmission is successful if and only if the received SINR is
    larger than a threshold determined by $R, W$, and $N$.
\end{itemize}
The second to last assumption should make it clear that we are
considering only an \textit{off-line} optimization of the frequency
band structure, and that no on-line (e.g., channel- and queue-based)
transmission or sub-band decisions are considered.

By considering such a network, tools from stochastic geometry can be
used to characterize the distribution of received interference and
thus to quantify the success probability of each transmission as a
function of the transmitter density and the SINR threshold.
In this context, the question at hand is determining the value of
$N$ that maximizes success probability for a given spatial density
of transmitters.
Rather than considering the optimization in terms of $N$, it is
convenient to pose the problem in terms of the {\em spectral
efficiency} of each communication $\frac{R}{W/N}$. Our main result
is an exact characterization of the optimal spectral efficiency in
the form of a simple fixed point equation.\footnote{Because SINR is
a function of spectral efficiency, this is equivalent to a
derivation of the optimal SINR threshold.} Furthermore, the optimal
spectral efficiency is seen to be a function only of the path-loss
exponent and the energy per information bit $\Eb = \frac{P}{N_0 R}$
(where $P$ is the received power, $N_0$ is the noise spectral
density, and $R$ is the rate \cite{verdu_wide}), and thus is
independent of the transmitter density.
In order for a network to operate optimally, $N$ should be increased
until the spectral efficiency $\left( \frac{NR}{W} \right)$ is equal
to its optimal value.

When thermal noise is negligible relative to the received signal
power (i.e., $\Eb \rightarrow \infty$), the network is purely
interference-limited and the optimal spectral efficiency is a
function of the path loss exponent ($\alpha$) alone.  For reasonable
path loss exponents the optimal spectral efficiency lies between the
low-SNR and high-SNR regimes.  For example, the optimal is $1.3$
bps/Hz (SINR threshold of $1.6$ dB) and $2.3$ bps/Hz (SINR threshold
of $5.9$ dB) for $\alpha=3$ and $\alpha=4$, respectively. When
thermal noise is not negligible (i.e., $\Eb$ is small), the optimal
spectral efficiency is shown to be the fraction $\left(1 -
\frac{2}{\alpha} \right)$ of the maximum spectral efficiency
achievable in the absence of interference.

Increasing $N$, which corresponds to decreasing the bandwidth and
increasing the area consumed by each transmission, is seen to be
beneficial as long as area (i.e., the SINR threshold) increases at a
reasonable rate with $N$. For interference-limited networks this is
true until the high-SNR regime is reached, at which point a huge
SINR increase is required for any additional bandwidth reduction.
For power-limited networks this is true until the SINR threshold
approaches the interference-free SNR, at which point the system
becomes overly sensitive to interference.

%%%%%%%%%%%%%%%%%%%%%%%%%%%%%%%%%%%%%%%%%%%%%%%%%%%%%%%%%%%%%%%%%
%%%%%%%%%%%%%%%%%%%%%%%%%%%%%%%%%%%%%%%%%%%%%%%%%%%%%%%%%%%%%%%%%
\subsection{Related Work}

The problem studied in this work is essentially the optimization of
frequency reuse in uncoordinated decentralized networks, which is a
well studied problem in the context of centrally-planned cellular
and other hierarchical networks; see for example
\cite{Rappaport,AndBook,YeuNan96} and references therein. In both
settings the tradeoff is between the bandwidth utilized per cell or
transmission -- which is inversely proportional to the frequency
reuse factor -- and the achieved SINR/spectral efficiency per
transmission. A key difference is that regular frequency reuse
patterns can be used in cellular networks, whereas in an ad hoc or
unlicensed network this is impossible. Another crucial difference is
in terms of analytical tractability. Although there has been a
tremendous amount of work on frequency reuse optimization for
cellular networks, these efforts generally do not lead to clean
analytical results.
On the contrary, in this work we are able to derive simple
analytical results for decentralized networks that cleanly show the
dependence of the optimal reuse factor on basic system parameters.

A number of works have considered related problems in the context of
decentralized networks, although none appear to have investigated
the optimization considered here.  In \cite{PursleyRoyster07} the
time-bandwidth-area product achieved by different codes are
evaluated.  This metric is essentially equivalent to the inverse of
transmission density in our network model, but the authors do not
pursue optimization of this metric, which is the essence of our
work. In \cite{EhsanCruz06}, the authors jointly optimize rate,
transmitter-receiver distance ($d$), and density in order to
maximize the transport capacity (i.e., product of rate and distance)
of a random-access network. This setting is very different from our
framework in which we assume a fixed rate and transmitter-receive
distance, and as a result conclusions differ significantly. For
example, the optimum SINR threshold in \cite{EhsanCruz06} for some
networks is found to be orders of magnitude smaller than 0 dB,
whereas we find optimal values around 0 dB. In
\cite{EbrahimiKhandani07} a network consisting of a large number of
interfering transmitter-receiver pairs is analyzed, but no spatial
model is used and only fading is considered. In
\cite{Sikora_Laneman_Haenggi} the issue of frequency reuse is
considered in a one-dimensional, evenly spaced, multi-hop wireless
network. Some similar general insights are derived, but the regular
spacing of interferers seems to prevent derivation of clean
analytical results as is possible for the 2-D network considered
here.  In a recent contribution the interactions between multiple
random-access networks have been considered from a game-theoretic
perspective \cite{GrokopTse07}, and portions of the analysis of a
single network in \cite{GrokopTse07} coincide with our initial
findings reported earlier in \cite{JinAndWeb_ITA07}.
 There has also been a good deal of work on multi-channel wireless networks, but this body of
 work generally deals with scheduled networks as opposed to our treatment of
unscheduled networks (see \cite{multichannel06} and references
therein).  Perhaps most relevant is \cite{Bahl07}, in which
algorithms for \textit{dynamic} allocation of bandwidth-area
resources are proposed.

%%%%%%%%%%%%%%%%%%%%%%%%%%%%%%%%%%%%%%%%%%%%%%%%%%%%%%%%%%%%%%%
\section{Preliminaries}

\subsection{Network Model}
We consider a set of transmitting nodes at an arbitrary snapshot in
time with locations specified by a homogeneous Poisson point process
(PPP) of intensity $\lambda$ on the infinite 2-D plane. All nodes
are assumed to simultaneously transmit with power $\rho$.
By the stationarity of the PPP it is sufficient to analyze the
behavior of a single reference TX-RX pair (TX 0, RX 0), separated by
assumption by a distance $d$.   Note that the receivers are not a
part of the transmitter process. From the perspective of RX 0, the
interferers follow the distribution of the PPP conditioned on the
location of TX 0 (referred to as the Palm distribution). However, by
Slivnyak's theorem \cite{StoKen96} this distribution is the same as
the unconditional distribution and therefore the locations of the
interfering nodes form a homogeneous PPP of intensity $\lambda$.
Received power is modeled by path loss with exponent $\alpha
> 2$.  If $X_i$ denotes the distance of the $i$-th transmitting node to the
reference receiver and the transmit signal of the $i$-th transmitter
is denoted as $U_i$, the reference received signal is:
\begin{eqnarray*}
Y_0 = U_0  d^{-\alpha/2} + \sum_{i \in \Pi(\lambda)} U_i
X_i^{-\alpha/2} + Z_i
\end{eqnarray*}
where $Z_i$ is additive Gaussian noise with power $\eta$.
The resulting SINR therefore is:
\begin{eqnarray*}
\SINR_0 = \frac{ \rho d^{-\alpha}}{ \eta + \sum_{i \in \Pi(\lambda)}
\rho X_i^{-\alpha}},
\end{eqnarray*}
where $\Pi(\lambda)$ indicates the point process describing the
(random) interferer locations.  If Gaussian signaling is used, the
received mutual information (conditioned on interferer locations) is
$I(U_0;Y_0 | \Pi(\lambda)) = \log_2 ( 1 + \SINR_0)$ bits/symbol.  In
the fixed rate setting considered here, the probability the received
mutual information is smaller than the transmission rate is known to
be a good approximation to packet error rate if strong channel
coding is used \cite{Caire_Biglieri}, and thus is the primary metric
in this work.

A few comments in justification of our model are in order. Although
the model contains many simplifications to allow for tractability,
it contains many of the critical elements of a real
 decentralized network.  First, the spatial Poisson distribution means
that transmitting nodes are randomly and independently located; this
is reasonable in a network with indiscriminate node placement or
substantial mobility
assuming that no intelligent transmission scheduling is performed
Scheduling generally attempts to ensure that simultaneous
transmissions are sufficiently separated in space, and thus can
significantly change the spatial distribution of simultaneous
transmissions.  However, even simple scheduling protocols can incur
considerable overhead and latency and thus unscheduled systems (or
systems using ALOHA-like protocols that make transmission decisions
independent of interference conditions) as considered here are of
interest.  This is particularly true when scheduling overhead begins
to overtake the advantage of scheduling, as may be the case with
high mobility or very bursty traffic. The assumptions of fixed TX-RX
distances and no fading are often not reasonable, but as we discuss
in Section \ref{sec-fading} our results also apply to networks with
fading and/or variable distances in the interference-limited regime
(no thermal noise).  Furthermore, our results are reasonably
accurate in the presence of non-negligible thermal noise when the
fading and distance variation is not too large.  Finally, we note
that fixed- rather than variable-rate communication is appropriate
for some, but not necessarily all, settings, e.g., single-hop
communication with very stringent delay constraints.  In other
settings (e.g., when delay constraints are less stringent) variable
rate communication is more appropriate; optimizing bandwidth
partitioning in this context is of interest but is outside the
scope of this work.

%%%%%%%%%%%%%%%%%%%%%%%%%%%%%%%%%%%%%%%%%%%%%%%%%%%%%%%%%%%%%%%55
\subsection{Outage Probability/Maximum Density Characterization} \label{sec-transcap}

An outage occurs whenever the SINR falls below threshold $\beta$, or
equivalently whenever the received mutual information is smaller
than $\log_2(1+ \beta)$. Therefore, the system-wide outage
probability is:
\begin{eqnarray*}
P_{\textrm{out}}(\lambda, \beta, \eta) \triangleq
 \mathbb{P} \left( \frac{ \rho d^{-\alpha}}{ \eta + \sum_{i \in
\Pi(\lambda)} \rho X_i^{-\alpha}} \leq \beta \right).
\end{eqnarray*}
This quantity is computed over the distribution of transmitter
positions and is an increasing function of the intensity $\lambda$.
The SINR threshold $\beta$ and the noise power $\eta$ are treated as
constants here, but are related to $R$, $W$, and $N$ in the
following section. Random variable $X$ is defined as the received
interference raised to the power $-\frac{2}{\alpha}$:
\begin{eqnarray*}
X \triangleq  \left( \sum_{i \in \Pi(\lambda)} X_i^{-\alpha}
\right)^{-\frac{2}{\alpha}},
\end{eqnarray*}
which allows the outage probability to be written in terms of $X$
as:
\begin{eqnarray*}
P_{\textrm{out}}(\lambda, \beta, \eta) &=&
 \mathbb{P} \left( \frac{ \rho d^{-\alpha}}{ \eta + \rho X^{-\frac{\alpha}{2}} } \leq \beta
 \right)
=  \mathbb{P} \left( X  \leq d^2 \left(\frac{1}{\beta} -
\frac{\eta}{\rho d^{-\alpha}} \right)^{-\frac{2}{\alpha}} \right).
\end{eqnarray*}
It is useful to write this expression in terms of a
\textit{normalized} interferer process.  If we define $Z$ as the
received interference for a process with intensity $\frac{1}{\pi}$:
\begin{eqnarray*}
Z \triangleq  \left( \sum_{i \in \Pi(1/\pi)} Z_i^{-\alpha}
\right)^{-\frac{2}{\alpha}},
\end{eqnarray*}
and note that a PPP with intensity $\lambda$ is equivalent to a PPP
with intensity $\frac{1}{\pi}$ scaled by $\frac{1}{\sqrt{\pi
\lambda}}$, it follows that $X$ and $\frac{1}{\pi \lambda} Z$ have
the same distribution.  Therefore
\begin{eqnarray}
\boxed{ P_{\textrm{out}}(\lambda, \beta, \eta) = F_Z \left( \lambda
\pi d^2 \left(\frac{1}{\beta} - \frac{\eta}{\rho d^{-\alpha}}
\right)^{-\frac{2}{\alpha}} \right) } \label{eq-outage_cdf}
%&=& \mathbb{P} \left( Z \leq
%\lambda \pi d^2  \left(\frac{1}{\beta} -
%\frac{\eta}{\rho d^{-\alpha}} \right)^{-\frac{2}{\alpha}} \right)\\
\end{eqnarray}
where $F_Z(\cdot)$ denotes the CDF of random variable $Z$. Although
a closed form expression for $F_Z(\cdot)$ is not known except for
the special case of $\alpha=4$ \cite{SouSil90}, this
characterization of the outage probability allows us to derive an
exact solution to the bandwidth partition problem.

In many scenarios, the network is subject to an outage constraint
and the quantity of interest is the maximum intensity of
\textit{attempted} transmissions $\lambda_{\epsilon}$ such that the
outage probability (for a fixed $\beta$) is no larger than
$\epsilon$. Because outage probability increases monotonically with
density, an expression for $\lambda_{\epsilon}$ is reached by
inverting (\ref{eq-outage_cdf}):
\begin{eqnarray} \label{eq-transcap}
\boxed{ \lambda_{\epsilon} = \frac{F_Z^{-1} \left( \epsilon
\right)}{\pi d^2 } \left(\frac{1}{\beta} - \frac{\eta}{\rho
d^{-\alpha}} \right)^{\frac{2}{\alpha}} }
\end{eqnarray}
where $F_Z^{-1}(\cdot)$ is the inverse of $F_Z(\cdot)$.

Because the SINR is upper bounded by only considering the
contribution of the \textit{nearest} interferer, a necessary (but
not sufficient) condition for successful communication is that a
circle centered about the receiver of area $\pi d^2 \left(
\frac{1}{\beta} - \frac{\eta}{\rho d^{-\alpha}}
\right)^{-\frac{2}{\alpha}}$ be free of interferers \cite{WebYan05}.
On the other hand, the \textit{effective area} consumed by each
transmission when an outage level of $\epsilon$ is required is the
inverse of the density $\lambda_{\epsilon}$:
\begin{eqnarray}
\frac{1}{\lambda_{\epsilon}} = \frac{1}{F_Z^{-1} \left( \epsilon
\right)} \pi d^2 \left(\frac{1}{\beta} - \frac{\eta}{\rho
d^{-\alpha}} \right)^{-\frac{2}{\alpha}},
\end{eqnarray}
which is the interferer-free area from the necessary condition above
multiplied by the constant $\frac{1}{F_Z^{-1} \left( \epsilon
\right)}$.  This constant factor, which increases without bound as
$\epsilon \rightarrow 0$ and which is larger than one for all but
the largest values of $\epsilon$, accounts for the fact that
transmitters are randomly located and can be intuitively thought of
as a back-off parameter that ensures the outage constraint is met.
This interpretation turns out to be useful when interpreting
bandwidth partitioning in terms of bandwidth and area.

\section{Problem Formulation and General Solution} \label{sec-opt_freq}

We are now able to address the problem of interest, which is
determining the number of sub-bands that maximize the density of
transmissions such that the outage probability is no larger than
$\epsilon$.  As made explicit at the end of this section, finding
the value of $N$ that minimizes outage probability for a fixed total
density of transmitters is the dual of this problem and has
precisely the same solution.  For the reader's reference, the
relevant system parameters are summarized in Table
\ref{table-parameters}.

\begin{table}
\centering
\begin{tabular}{|l|l|}
\hline
Parameter & Description \\ \hline
$R$ & Transmission Rate (bits/sec) \\ \hline
$W$ & Total System Bandwidth (Hz) \\ \hline
$\rho$ & Transmission Power \\ \hline
$N_0$ & Noise Spectral Density \\ \hline
$d$ & Transmitter-Receiver Distance \\ \hline
$\Eb = \frac{\rho d^{-\alpha}}{N_0 R}$ & Energy per Information Bit \\ \hline
$\epsilon$ & Outage Constraint \\ \hline
$N$ & Number of Sub-Bands \\ \hline
$\beta$ & SINR Threshold \\ \hline
\end{tabular}
\caption{Summary of System Parameters}
\label{table-parameters}
\end{table}

If the system bandwidth is not split ($N=1$), each node utilizes the
entire bandwidth of $W$ Hz. The SINR required ($\beta$) to achieve a
rate of $R$ bps is determined by inverting the AWGN capacity
expression $R = W \log_2 (1 + \beta)$, which gives $\beta =
2^{\frac{R}{W}} - 1$. The maximum intensity of transmissions can be
determined by evaluating (\ref{eq-transcap}) with this value of
$\beta$ and $\eta = N_0 W$. If the system bandwidth is split
into $N>1$ orthogonal sub-bands each of width $\frac{W}{N}$, and
each transmitter-receiver pair uses one \textit{randomly selected}
sub-band, the required SINR $\beta(N)$ is determined by inverting
the rate expression:
\begin{eqnarray} \label{eq-beta}
R &=& \frac{W}{N} \log_2 (1 + \beta(N)) ~~~ \rightarrow ~~~
%\end{eqnarray*}
%which yields
%\begin{eqnarray}
\beta(N) = 2^{\frac{N R}{W}} - 1.
\end{eqnarray}
Because each transmitter randomly chooses a sub-band, the users on
each sub-band are still a PPP and are independent across bands.  As
a result, the maximum intensity of transmissions \textit{per
sub-band} is $\lambda_{\epsilon}$ as defined in (\ref{eq-transcap})
with SINR threshold $\beta(N)$ and noise power $\eta = N_0
\frac{W}{N}$. Since the $N$ sub-bands are statistically identical,
the maximum total intensity of transmissions, denoted
$\lambda^T_{\epsilon}$, is the per sub-band intensity
$\lambda_{\epsilon}$ multiplied by $N$. Therefore, from
(\ref{eq-transcap}) we have:
\begin{eqnarray} \label{eq-tc_approx1}
\lambda_{\epsilon}^T(N) = N \left( \frac{F_Z^{-1} \left(\epsilon
\right)}{\pi d^2 } \right) \left( \frac{1}{\beta(N)} - \frac{N_0
\left(\frac{W}{N}\right)}{\rho d^{-\alpha}}
\right)^{\frac{2}{\alpha}}.
\end{eqnarray}
%where the constant  $\snr \triangleq \frac{\rho d^{-\alpha}}{N_0 W}$
%is the received signal-to-noise ratio in the absence of interference when the
%entire band is used.
The optimal number of sub-bands $N^*$ is that which maximizes
total transmission density:
\begin{eqnarray} \label{eq-optN}
N^* = \textrm{arg} \max_N \lambda_{\epsilon}^T(N).
\end{eqnarray}
It is useful to interpret this optimization in terms of bandwidth
and area. Dividing (\ref{eq-tc_approx1}) by $W$ and then inverting
yields:
\begin{eqnarray}  \label{eq-tradeoff}
\frac{W}{\lambda_{\epsilon}^T(b)} = \frac{1}{F_Z^{-1} \left(
\epsilon \right)}
 \underbrace{\left(\frac{W}{N}\right)}_{\textrm{Bandwidth}}
\underbrace{\pi d^2 \left( \frac{1}{\beta(N)} - \frac{N_0 W}{N \rho
d^{-\alpha}} \right)^{-\frac{2}{\alpha}}}_{\textrm{Interferer-Free
Area}}.
\end{eqnarray}
which is the product of the constant $\frac{1}{F_Z^{-1} \left(
\epsilon \right)}$, sub-band bandwidth $\frac{W}{N}$, and the
required interferer-free area.  Total density is maximized by
minimizing this quantity, i.e., by minimizing the
\textit{bandwidth-area product} of each transmission. It is easily
checked that the interferer-free area is a strictly increasing
function of $N$. Thus, as the number of sub-bands $N$ is increased
the bandwidth consumed by each transmission decreases while the area
increases, leading to a non-trivial tradeoff.

Rather than solving the maximization in (\ref{eq-optN}) with respect
to $N$, it is more convenient to maximize with respect to the
\textit{operating spectral efficiency}, which is equal to the
transmission rate divided by the bandwidth of each sub-band:
\begin{equation}
b \triangleq \frac{R}{W/N} ~ \textrm{bps/Hz}.
\end{equation}
It is important to note that the operating spectral efficiency $b$
is a design parameter even though the per-transmission rate $R$ and
system bandwidth $W$ are fixed.\footnote{If only bandwidth
optimization is considered, $b$ should be limited to integer
multiples of $\frac{R}{W}$; in this case $N^*$ is either the integer
floor or ceiling of $b^* \left(\frac{R}{W}\right)$ due to the nature
of the objective function. However, if a more general scenario is
considered where the sub-band structure as well as the length of
transmission is being designed (e.g., in a packetized system), these
two parameters allow for operation at any desired $b$. Therefore,
arbitrary $b > 0$ are considered for the remainder of the paper.}

With this substitution the transmission density can be written as a
function of $b$:
\begin{eqnarray}
\lambda_{\epsilon}^T(b) = \left(
\frac{F_Z^{-1} \left( \epsilon\right)}{\pi d^2 } \right) 
\left(\frac{W}{R}\right)  b \left(
\frac{1}{2^b-1} - \frac{1}{b}\frac{N_0 R}{\rho d^{-\alpha}}
\right)^{\frac{2}{\alpha}}
%+ \Theta(\epsilon^2),
\end{eqnarray}
Noting that the constant $\frac{\rho d^{-\alpha}}{N_0 R}  \triangleq
\Eb$ is the received energy per information bit \cite{verdu_wide}
and defining the constant $\kappa \triangleq
\left( \frac{F_Z^{-1} \left(\epsilon\right)}{\pi d^2} \right) \left(\frac{W}{R}\right) $, this can be further simplified
as:
\begin{eqnarray}
\lambda_{\epsilon}^T(b)
%= \left(\frac{W}{R}\right) \left(
%\frac{-\log_e(1-\epsilon)}{\pi d^2} \right) b \left( \frac{1}{2^b-1} -
%\frac{1}{\Eb b} \right)^{\frac{2}{\alpha}}
= \kappa b \left(\frac{1}{2^b-1} - \frac{1}{b \Eb}
\right)^{\frac{2}{\alpha}}.
\end{eqnarray}
The optimal spectral efficiency $b^*$ is therefore the solution to the
following optimization:
% The optimization we need to solve therefore is:
\begin{equation} \label{eq-master_b}
b^* =\textrm{arg} \max_{b>0} ~~b \left( \frac{1}{2^b-1} - \frac{1}{b \Eb} \right)^{\frac{2}{\alpha}}.
%= \textrm{arg} \max_{b>0} \lambda_{\epsilon}^T (b)
\end{equation}
Note that the optimal $b^*$ depends only on the path loss exponent
$\alpha$ and $\Eb$, and thus any dependence on power and rate is
completely captured by $\Eb$. By posing the problem in terms of
spectral efficiency, any direct dependence on $W$ is removed.
Furthermore, the problem is completely independent of the outage
constraint $\epsilon$.

The problem in (\ref{eq-master_b}) is only feasible for $b$ satisfying $\frac{1}{2^b-1} -
\frac{1}{b \Eb} \geq 0$, which corresponds to the SINR threshold $\beta = 2^b-1$ being
no larger than the interference-free SNR $\frac{N \rho d^{-\alpha}}{N_0 W}$.
Some simple manipulation shows that
this condition is equivalent to $b \leq \Ce \Ebp $, where $\Ce \Ebp$
is the maximum spectral efficiency of an AWGN channel and thus
is the solution to \cite[Equation 23]{verdu_wide}:
\begin{eqnarray} \label{eq-bmax}
%\frac{2^{\Ce \Ebp} - 1}{\Ce \Ebp} = \Eb.
2^{\Ce \Ebp} - 1 = \Eb \Ce
\Ebp.
\end{eqnarray}
The domain of the maximization is thus $0 \leq b \leq \Ce \Ebp$. If
$\Eb \leq \log_e 2 = -1.59$ dB the problem is infeasible for any $b$
because this corresponds to operating beyond interference-free
capacity\footnote{For readers less familiar with the power-limited
regime, note that fixing power $P$ and noise spectral density $N_0$
and using less bandwidth leads to a decreasing rate, i.e., the
function $w \log_2 \left(1 + \frac{P}{N_0 w} \right) \downarrow 0$
 as $w \rightarrow 0$.  Thus, there is a minimum bandwidth needed
to achieve a particular rate $R$ even in the absence of multi-user interference;
 this is the solution to
$R = w \log_2 \left(1 + \frac{P}{N_0 w} \right)$ and is precisely
the quantity $\frac{R}{\Ce \Ebp}$.  Furthermore, note that $w \log_2
\left(1 + \frac{P}{N_0 w} \right) \uparrow
\left(\frac{P}{N_0}\right) \log_e 2$
 as $w \rightarrow \infty$; therefore the minimum energy
per information bit $\Eb_{\textrm{min}} = \frac{P}{N_0 R} = \log_e 2
= -1.59$ dB and $\Ce (\log_e 2)=0$. }.

By taking the derivative of $\lambda_{\epsilon}^T(b)$ and setting it
equal to zero, the optimal spectral efficiency $b^*$ can be
characterized in terms of a fixed point equation parameterized by
$\alpha$ and $\Eb$:
\begin{theorem} \label{thm-main}
The optimum operating spectral efficiency $b^*$ is the \textit{unique} positive
solution of the following equation:
\begin{equation} \label{eq-main}
\boxed{ \Eb b \left(2^b-1 \right) - \Eb \frac{2}{\alpha} b^2 2^b \log_e 2 -
\left(1- \frac{2}{\alpha} \right) \left(2^b-1 \right)^2 = 0 }
\end{equation}
Furthermore, $b^*$ is an increasing function of $\Eb$ and of
$\alpha$.
\end{theorem}
\vspace{2mm}
\begin{proof}
See Appendix \ref{app-proof-main}.
\end{proof}

Although we are not able to find a general closed-form expression
for (\ref{eq-main}), this expression is easily solved numerically
and we can find closed form solutions in the asymptotic regimes
($\Eb \rightarrow \infty$ and $\Eb \rightarrow -1.59$ dB). In Fig.
\ref{fig-alpha4} the numerically computed optimum spectral
efficiency $b^*$ and the corresponding density constant
$\frac{\lambda_{\epsilon}^T \left( b^* \right)}{\kappa}$ 
are plotted versus $\Eb$ for $\alpha=4$, along with the
spectral efficiency of an interference-free AWGN channel $\Ce \Ebp$.
From this figure, two asymptotic regimes of interest can be
identified:

%{fig-alpha4} here
\begin{figure}
\centering
\includegraphics[width=4in]{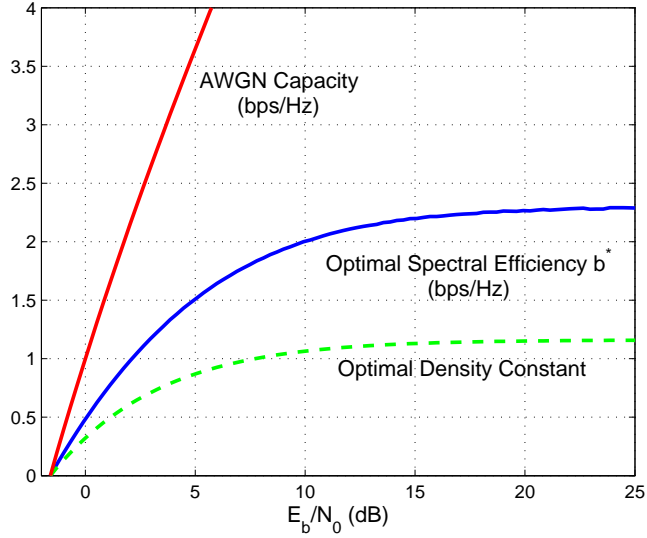}
\caption{Optimal Spectral Efficiency $b^*$ and Optimal Density
Constant $\frac{\lambda_{\epsilon}^T \left( b^* \right)}{\kappa}$
vs. $\Eb$ for $\alpha=4$.} \label{fig-alpha4}
\end{figure}

\begin{itemize}
\item {\bf Interference-Limited Networks:} When $\Eb$ is sufficiently
large, the effect of thermal noise vanishes and performance depends only
on multi-user interference.  As a result, the optimal $b^*$ and 
density $\lambda_{\epsilon}^T \left( b^* \right)$ both converge to
constants as $\Eb \rightarrow \infty$.

\item {\bf Power-Limited Networks:} When $\Eb$ is close to its minimum value of
$-1.59$ dB, $b^*$ and $\lambda_{\epsilon}^T(b^*)$ scale linearly
with $\Eb$ (dB) and show characteristics very similar to AWGN
spectral efficiency \cite{verdu_wide}.
%In this regime performance is determined by
%interference and thermal noise, and $b^*$ and $\lambda_{\epsilon}^T
%\left( b^* \right)$ go to zero as $\Eb$ decreases towards $-1.59$
%dB.
\end{itemize}

In Section \ref{sec-int_limited} the interference-limited regime is
explored and a closed form expression for the optimal value of $b^*$
in terms of only the path-loss exponent is derived.  Once a system
is in this regime, performance is virtually unaffected by further
increasing transmission power.
%Note that many systems of interest (e.g., WiFi) operate in this regime, and this
%regime is where most previous research on frequency reuse has
%concentrated.
In Section \ref{sec-wideband} the power-limited regime is explored
and simple expressions for $b^*$ and $\lambda_{\epsilon}^T(b^*)$ in
terms of $\alpha$ and $\Eb$ are given that are accurate for $\Eb$
near $-1.59$ dB. Although intuition might suggest that noise is
dominant and thus interference is negligible in this regime, this is
not the case as evidenced by the fact that the optimum spectral
efficiency $b^*$ is considerably smaller than the interference-free
spectral efficiency $\Ce \Ebp$.  Furthermore, increasing
transmission power does significantly increase density in this
regime. Between these two regimes (approximately from 2-3 dB to
15-20 dB), $b^*$ increases sub-linearly with $\Eb$ (dB) and the
intuition is a combination of the insights derived for the
interference- and power-limited regimes.

In Fig. \ref{fig-opt_spec3} numerically computed values of $b^*$ are
plotted versus $\Eb$ for $\alpha=2.5,~3,~3.5$ and $4$, and the
interference-limited regime is seen to begin around $15$ dB for each
value of $\alpha$.  Although not visible here, it is interesting to
note that $\frac{\lambda_{\epsilon}^T(b^*)}{\kappa}$ is not monotonic
with respect to $\alpha$; on the other hand,  it is easily verified that
$\lambda_{\epsilon}^T(b^*)$ monotonically increases with $\Eb$.

\begin{figure}
\centering
\includegraphics[width=4in]{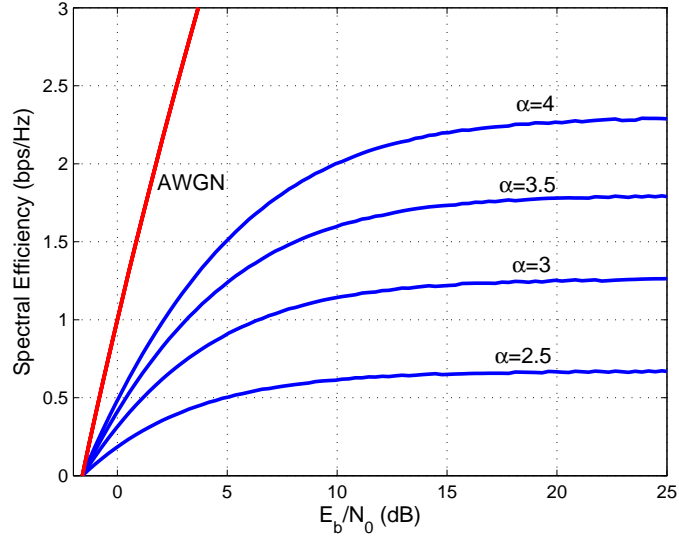}
\caption{Optimal Spectral Efficiency $b^*$ vs. $\Eb$ for
$\alpha=2.5, 3, 3.5, 4$.} \label{fig-opt_spec3}
\end{figure}

\begin{remark} \label{remark-outage}
The dual problem of density maximization subject to an outage
constraint is outage minimization for a given density. In this case
the overall outage probability is the same as the outage probability
on each of the $N$ sub-bands, each of which has density
$\frac{\lambda}{N}$. Substituting appropriate values for the SINR
threshold and the noise power in (\ref{eq-outage_cdf}) yields:
\begin{eqnarray}
 P_{\textrm{out}}(b) = F_Z \left( \lambda \pi d^2
 \left(\frac{R}{W}\right) \left( \frac{1}{b} \right)
 \left(\frac{1}{2^b-1} - \frac{1}{b \Eb} \right)^{-\frac{2}{\alpha}}
 \right).
\end{eqnarray}
Outage probability is minimized by minimizing the argument of the
CDF (due to the non-decreasing nature of any CDF).  Because the
argument is inversely proportional to the argument of the
maximization in (\ref{eq-master_b}), the problems of outage
minimization and density maximization are equivalent.  To understand
the impact of partitioning, it is useful to note that $F_Z(z)$ is
approximately linear for small $z$ \cite{WebYan05}.
 \hfill $\lozenge$
\end{remark}

\begin{remark}
If the available transmission rates are at a \emph{gap} to capacity,
i.e., $R = W \log_2 (1 + \Gamma^{-1} \cdot \SINR)$ for some $\Gamma
> 1$, the required SINR increases by a factor of $\Gamma$ to
$\beta(N) = \Gamma \left(2^{NR/W} - 1 \right)$ and the density is
given by $\lambda_{\epsilon}^T(b) = \Gamma^{-\frac{2}{\alpha}}
\kappa b \left(\frac{1}{2^b-1} - \frac{1}{b \Eb}
\right)^{\frac{2}{\alpha}}$ where $\Eb = \frac{\rho
d^{-\alpha}}{\Gamma N_0 R}$.  Thus, the optimal spectral efficiency
is given by evaluating Theorem \ref{thm-main} with $\Eb = \frac{\rho
d^{-\alpha}}{\Gamma N_0 R}$.
% and the corresponding transmission
%density is decreased by a factor of $\Gamma^{-\frac{2}{\alpha}}$.
\hfill $\lozenge$
\end{remark}

%\begin{remark}
%Although interference is not summable when $\alpha=2$, it is
%interesting to note that the objective of (\ref{eq-master_b}) is
%decreasing in $b$ if $\alpha=2$ and therefore the optimal value of
%$b$ is zero.  However, Theorem \ref{thm-main} implies that the
%optimal value of $b$ is strictly greater than zero for any $\alpha >
%2$.  \hfill $\lozenge$
%\end{remark}

%%%%%%%%%%%%%%%%%%%%%%%%%%%%%%%%%%%%%%%%%%%%%%%%%%%%%%%%%%%%%%%%%%%
%%%%%%%%%%%%%%%%%%%%%%%%%%%%%%%%%%%%%%%%%%%%%%%%%%%%%%%%%%%%%%%%%
\section{Partitioning for Interference-Limited Networks} \label{sec-int_limited}

In systems with sufficiently powered devices (i.e., large $\Eb$),
thermal noise is essentially negligible.  In the limiting case where
$N_0 = 0$ (i.e., $\Eb \rightarrow \infty$) the density is given by:
\begin{equation} \label{eq-int_limited_b}
\lambda_{\epsilon}^T(b) = \kappa b \left(2^b-1 \right)^{-
\frac{2}{\alpha}}.
\end{equation}
In this limiting regime, a closed-form solution for $b^*$ can be reached.
\begin{theorem} \label{thm-int-limited}
The optimum operating spectral efficiency $b^*$ in the absence of
thermal noise ($N_0 = 0 \leftrightarrow \Eb = \infty$) is the unique
solution to:
\begin{eqnarray} \label{eq-optspec}
b^* = (\log_2 e) \frac{\alpha}{2} (1 - 2^{-b^*}),
\end{eqnarray}
which can be written in closed form as:
\begin{equation} \label{eq-optspec2}
\boxed{ b^* = \log_2 e \left[ \frac{\alpha}{2} + \mathcal{W} \left( -
\frac{\alpha}{2} e^{-\frac{\alpha}{2}} \right) \right] }
\end{equation}
where $\mathcal{W}(z)$ is the principle branch of the Lambert $\mathcal{W}$ function
and thus solves $\mathcal{W}(z) e^{ \mathcal{W}(z)} = z$.
\end{theorem}
\vspace{2mm}
\begin{proof}
The result can be shown by directly maximizing (\ref{eq-int_limited_b})
or by solving the fixed point equation given in Theorem \ref{thm-main}
while keeping only the $\Eb$ terms.
The latter approach yields:
\begin{equation*}
\frac{2}{\alpha} b^2 2^b \log_e 2 - b \left(2^b-1 \right) = 0,
\end{equation*}
which is easily manipulated into the form of (\ref{eq-optspec}).  To
get (\ref{eq-optspec2}) we manipulate (\ref{eq-optspec}) into the
form $\left(b^* \log_e 2 - \frac{\alpha}{2} \right) e^{b^* \log_e 2} = -\frac{\alpha}{2}$.
Multiplying both sides by $e^{-\frac{\alpha}{2}}$ yields
$\left(b^* \log_e 2 - \frac{\alpha}{2} \right) e^{b^* \log_e 2 - \frac{\alpha}{2}}
= -\frac{\alpha}{2} e^{-\frac{\alpha}{2}}$, from which we have
$\mathcal{W} \left( -\frac{\alpha}{2} e^{-\frac{\alpha}{2}} \right) =
b^* \log_e 2 - \frac{\alpha}{2}$ and thus the result.
\end{proof}

The optimum depends only on the path loss exponent $\alpha$, and it
is straightforward to show that $b^*$ is an increasing function of
$\alpha$, $b^*$ is upper bounded by $\frac{\alpha}{2} \log_2 e$, and
that $b^*/(\frac{\alpha}{2} \log_2 e)$ converges to $1$ as $\alpha \rightarrow \infty$. 
In Fig. \ref{fig-opt_spec} the optimal spectral
efficiency $b^*$ and $\frac{\lambda_{\epsilon}^T(b^*)}{\kappa}$ are plotted versus path-loss exponent $\alpha$.
The optimal spectral efficiency is very small for $\alpha$ close to
2 but then increases nearly linearly with $\alpha$; for example, the
optimal spectral efficiency for $\alpha=3$ is $1.26$ bps/Hz
($\beta=1.45$ dB).  Note the non-monotonic behavior of 
$\frac{\lambda_{\epsilon}^T(b^*)}{\kappa}$ with $\alpha$: the minimum
occurs at $\alpha=2.77$, where
$\frac{\lambda_{\epsilon}^T(b^*)}{\kappa} = b^* = 1$.

%{fig-opt_spec} here

\begin{figure}
\centering
\includegraphics[width=4in]{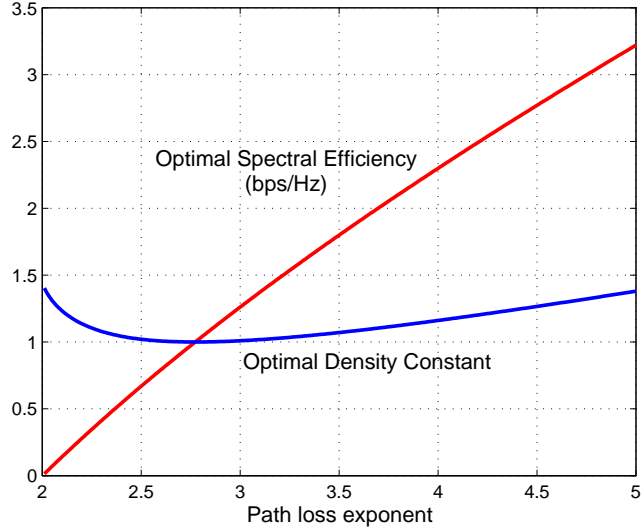}
\caption{Optimal Spectral Efficiency $b^*$ and Optimal Density
Constant $\frac{\lambda_{\epsilon}^T \left( b^* \right)}{\kappa}$
vs. Path Loss Exponent $\alpha$ for Interference-Limited Networks.}
\label{fig-opt_spec}
\end{figure}

To gain an intuitive understanding of the optimal solution, let us
first consider the behavior of $\lambda_{\epsilon}^T(b)$ when $b$ is
small, i.e. $b \ll 1$.  Because $e^x - 1 \approx x$ for small $x$,
the SINR threshold increasing approximately \textit{linearly} with
$b$: $\beta = 2^b-1 \approx b \log_e 2$. Plugging into
(\ref{eq-int_limited_b}) yields:
\begin{eqnarray*}
\lambda_{\epsilon}^T(b) = \kappa b \left(2^b-1 \right)^{-
\frac{2}{\alpha}} \approx \kappa b \cdot b^{- \frac{2}{\alpha}} =
\kappa b^{\left(1 - \frac{2}{\alpha}\right)}.
\end{eqnarray*}
For any path-loss exponent $\alpha > 2$, the density of
transmissions increases as $b^{\left(1-\frac{2}{\alpha} \right)}$.
Therefore, \textit{increasing the number of sub-bands $N$, or
equivalently increasing the spectral efficiency $b$, leads to an
increased transmission capacity}, as long as the linear
approximation to $\beta$ remains valid.  Recall that the area
consumed by each transmission is proportional to
$\beta^{\frac{2}{\alpha}}$ (equation \ref{eq-tradeoff}): if $\beta
\sim b$, then area increases sub-linearly as $b^{\frac{2}{\alpha}}$
and this increase is offset by the linear increase in the number of
parallel transmissions. When $b$ becomes larger, $\beta$ begins to
grow \textit{exponentially} rather than linearly with $b$ (i.e.,
SINR must be doubled \textit{in dB units} rather than in linear
units in order to double spectral efficiency) and thus the benefit
of further increasing the number of sub-bands is far outweighed by
the SINR/area increase.

% In fact, doubling spectral efficiency requires
%doubling the SINR . For example, the required SINR for $b=0.5$ is
%approximately 3.4 dB higher than that required for $b=0.25$ ($-3.8$
%dB versus $-7.2$ dB); on the other hand, the required SINR for $b=4$
%is $11.8$ dB versus $4.8$ dB for $b=2$.

This behavior is illustrated in Fig. \ref{fig-density_vs_spec},
where $\lambda_{\epsilon}^T(b) = \kappa b \left(2^b-1 \right)^{-
\frac{2}{\alpha}}$ (with $\kappa=1$) is plotted versus $b$ for
different values of $\alpha$. The function increases rapidly when
$b$ is small, but then decreases rapidly beyond its peak when the
SINR cost becomes prohibitive.  A larger path loss exponent makes
the system less sensitive to interference, and thus the peak is
attained at a larger value of $b$.  It is interesting to note that
all of the curves intersect at $b=1$ because
$\lambda_{\epsilon}^T(1) = \kappa$ for any value of $\alpha$.
Although $b=1$ is quite sub-optimal when $\alpha$ is near $2$,
$\kappa$ is reasonably close to the optimal $\kappa
b^*(2^{b^*}-1)^{- \frac{2}{\alpha}}$ for exponents between $2.5$ and
$5$ and thus is a rather robust operating point if the path loss
exponent is not known exactly.

\begin{figure}
\centering
\includegraphics[width=4in]{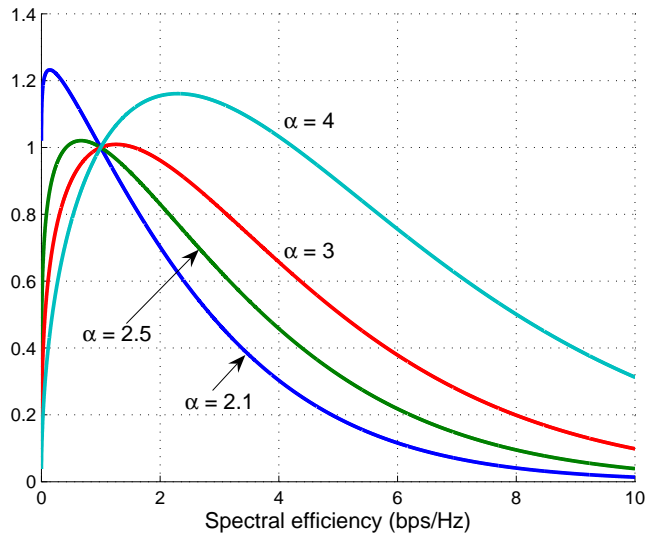}
\caption{Density Constant $\frac{\lambda_{\epsilon}^T (b)}{\kappa}$
vs. Spectral Efficiency $b$ for Interference-Limited Networks,
$\alpha = 2.1, ~2.5, 3~, 4$.} \label{fig-density_vs_spec}
\end{figure}

\textbf{A Design Example.}  Consider wireless LAN parameters that
are conceptually similar to those of an 2.4 GHz 802.11 system, that
uses $N=3$ bands of about $20$ MHz. Assume the usable bandwidth is a
total of $W = 60$ MHz, and that the desired rate is $R = 10$ Mbps
and $\alpha = 3$. From Theorem \ref{thm-int-limited} we can
determine that
\begin{equation}
N^* = \frac{b^*}{R/W} = \frac{1.26}{R/W} = 7.56,
\end{equation}
so the optimum partition is about $N^*=8$, or bands of 7.5 MHz. If
however the data rate requirement is higher, like $60$ Mbps, then it
can quickly be confirmed that $N^* = 1$.  That is, the maximum
number of users can be accommodated at the higher data rate if each
of them uses the entire band, since they
can accept a lower received SINR with such a large bandwidth. \\

%%%%%%%%%%%%%%%%%%%%%%%%%%%%%%%%%%%%%%%%%%%%%%%%%%%%%%%%%%%%%%%%%%%
%%%%%%%%%%%%%%%%%%%%%%%%%%%%%%%%%%%%%%%%%%%%%%%%%%%%%%%%%%%%%%%%%%%%%
\section{Partitioning for Power-Limited Networks} \label{sec-wideband}

In the power-limited regime where $\Eb$ is close to $-1.59$ dB, we
can obtain a simple characterization of $b^*$ that is accurate up to
a quadratic term by solving the fixed point equation given in
Theorem \ref{thm-main}:
\begin{theorem} \label{thm-wideband}
The optimum operating spectral efficiency $b^*$ in the power-limited
regime ($\Eb$ slightly larger than $-1.59$ dB) is given by:
\begin{equation} \label{eq-optspec_wide}
\boxed{ b^* = \left(1 - \frac{2}{\alpha} \right) \Ce \Ebp
+ O \left( b^2 \right)}
\end{equation}
where $ \Ce \Ebp$ is the AWGN spectral efficiency at $\Eb$ as defined in (\ref{eq-bmax}).

Furthermore, the density in the wideband regime is characterized as:
\begin{equation} \label{eq-lambda_wide}
\frac{\lambda_{\epsilon}^T(b^*)}{\kappa} = \left(
(1-\delta)^{(1-\delta)} \delta^{\delta} 2^{-\delta} \right) \Ce \Ebp
+ O \left( b^2 \right)
\end{equation}
where $\delta \triangleq \frac{2}{\alpha}$ and
$(1-\delta)^{(1-\delta)} \delta^{\delta} 2^{-\delta} < 1$ for all $\alpha > 2$.
\end{theorem}
\begin{proof}
See Appendix \ref{app-wideband}.
\end{proof}

Fig. \ref{fig-opt_spec_wideband} contains plots of the numerically
computed $b^*$, the approximation $(1-\frac{2}{\alpha}) \Ce \Ebp$,
and $\Ce \Ebp$ versus $\Eb$ for $\alpha=3$ and $\alpha=4$. Fig.
\ref{fig-opt_ase_wideband} contains plots of the numerically
computed $\frac{\lambda_{\epsilon}^T(b^*)}{\kappa}$, the
approximation from (\ref{eq-lambda_wide}), and $\Ce \Ebp$ versus
$\Eb$ for $\alpha=2.01$ and $\alpha=3$ (the curve for $\alpha=4$ is
nearly indistinguishable from $\alpha=3$). Both approximations are
seen to be very accurate.

\begin{figure}
\centering
\includegraphics[width=4in]{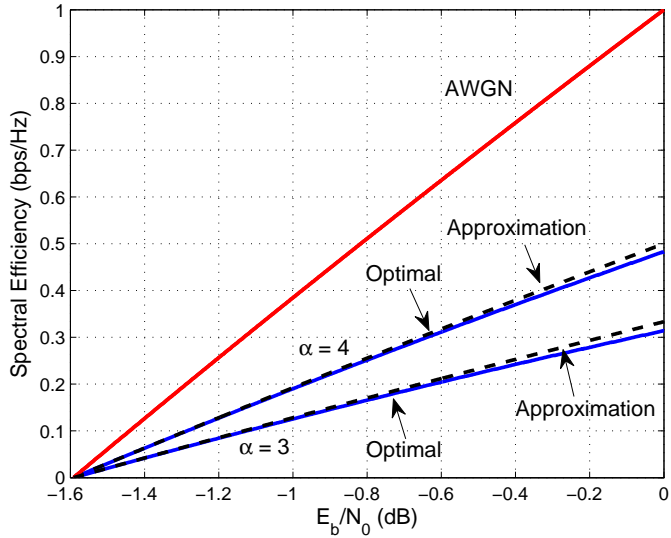}
\caption{Optimal Spectral Efficiency $b^*$ vs. $\Eb$ for
Power-Limited Networks, $\alpha=3,~4$.}
\label{fig-opt_spec_wideband}
\end{figure}

\begin{figure}
\centering
\includegraphics[width=4in]{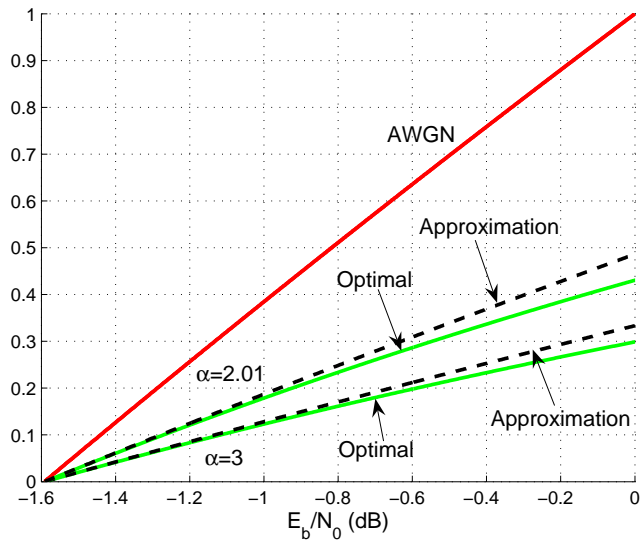}
\caption{Optimal Density Constant $\frac{\lambda_{\epsilon}^T
(b^*)}{\kappa}$ vs. $\Eb$ for Power-Limited Networks,
$\alpha=2.01,~3$.} \label{fig-opt_ase_wideband}
\end{figure}

Although intuition might suggest that interference can be ignored
when thermal noise is so large, this is not the case. If $b$ is
chosen only slightly smaller than $\Ce \Ebp$, the SINR threshold is
almost equal to the interference-free SNR and thus each receiver is
extremely sensitive to interference.  As a result each communication
consumes a large area, and this offsets the bandwidth savings of
using a large $b$.
%If $b=\Ce \Ebp$ then no interference
%can be handled and each communication consumes infinite area, which
%leads to $\lambda_{\epsilon}^T\left( \Ce \Ebp \right) = 0$;
%furthermore $\lambda_{\epsilon}^T(b) \rightarrow 0$ as $b
%\rightarrow \Ce \Ebp$.
On the other extreme, small $b$ corresponds to a small area because
the SINR threshold is much smaller than the interference-free SNR,
but this is offset by a large bandwidth which causes
$\lambda_{\epsilon}^T(b) \rightarrow 0$ as $b \rightarrow 0$.

This behavior is illustrated in Fig. \ref{fig-lambda_wideband},
where $\lambda_{\epsilon}^T(b)$ (with $\kappa=1$) is plotted versus
$b$ for $\alpha=2.2, 3$, and $4$ at $\Eb = -0.82$ dB (for which $\Ce
\Ebp = 0.5$ bps/Hz).  Choosing $b$ near either extreme leads to very
poor performance for any $\alpha$. Notice that all three curves
intersect when $b$ satisfies $\frac{1}{2^b-1} - \frac{1}{b \Eb} =
1$.  This condition is satisfied when the SINR threshold is equal to
$\frac{ \textrm{SNR}}{1 + \textrm{SNR}}$, where $\textrm{SNR}$ is
the interference-free SNR, and thus $b = \log_2 \left( 1 + \frac{
\textrm{SNR} }{1 + \textrm{SNR}} \right)$. By a simple calculation
using tools from \cite{verdu_wide}, the intersection point
corresponds to $b = \frac{1}{3} \Ce \Ebp$.  Although this choice of
spectral efficiency is only optimal for $\alpha=3$, it is quite
close to optimal for path loss exponents that are not too near $2$
and thus is a robust operating point in the power-limited regime,
analogous to the choice $b=1$ in the interference-limited regime.

\begin{figure}
\centering
\includegraphics[width=4in]{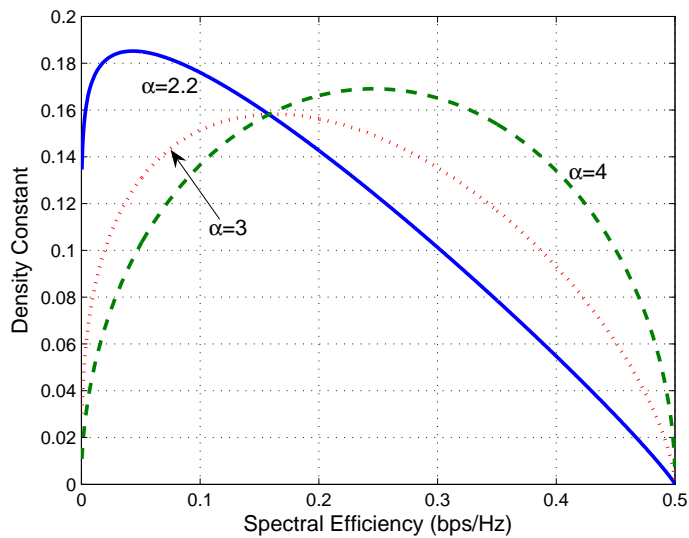}
\caption{Density Constant $\frac{\lambda_{\epsilon}^T (b)}{\kappa}$
vs. Spectral Efficiency $b$ for Power-Limited Networks.}
\label{fig-lambda_wideband}
\end{figure}

Finally, note that multi-user interference decreases the marginal
benefit of increased power (i.e., $\Eb$) as compared to an
interference-free channel. The analogous quantity for the spatial
network considered here is the \textit{area spectral efficiency}
(ASE), which is computed relative to the total bandwidth of $W$ Hz
and is equal to $\lambda_{\epsilon}^T(b^*) \left(\frac{R}{W}\right)$
bps/Hz per $m^2$.  In an AWGN channel, spectral efficiency increases
at a slope of 2 bps/Hz per 3 dB in the wideband regime
($\mathcal{S}_0=2$) \cite{verdu_wide}, while (\ref{eq-lambda_wide})
implies that ASE increases only at a rate of $2^{1-\delta}  \left(
(1-\delta)^{(1-\delta)} \delta^{\delta} \right)$ ($<2$) bps/Hz per 3
dB.

%For example, ASE increases at rates $\frac{2}{3}$ and
%$\frac{1}{\sqrt{2}}$ bps/Hz per 3 dB when $\alpha=3,~4$,
%respectively.

%{fig-opt_spec_wideband}
% {fig-opt_ase_wideband}

%%%%%%%%%%%%%%%%%%%%%%%%%%%%%%%%%%%%%%%%%%%%%%%%%%%%%%%%%%%%%%%%%%%%%%
%%%%%%%%%%%%%%%%%%%%%%%%%%%%%%%%%%%%%%%%%%%%%%%%%%%%%%%%%%%%%%%%%%%%%%
\section{Numerical Results and Extensions}

In the following we present numerical results to illustrate the
value of bandwidth partitioning. With system parameters chosen as
$\epsilon=0.1$, $N_0 = 10^{-6}$, $\alpha=4$, $d=10$, $R = 1$ Mbps,
and $W=10$ MHz, the total density $\lambda_{\epsilon}^T(N)$ is
computed via full Monte Carlo simulation (of outage probability at
different densities) and with equation (\ref{eq-tc_approx1}) using
the numerically computed value $F_Z^{-1}(0.1) = 0.1015$. Fig.
\ref{fig-numerical} contains plots of both quantities for $\Eb$
equal to $30$, $20$, $5$, and $0$ dB, and the curves match almost
exactly with any difference due purely to simulation error.
According to the chosen parameters we have $N^* = \frac{b^*}{R/W} =
10 b^*$ and $\kappa=0.0032$.  Note that the optimizing spectral
efficiency $b^*$ and the value of $\frac{\lambda_{\epsilon}^T \left(
b^* \right)}{\kappa}$ can be read from Fig. \ref{fig-alpha4}. The
top two set of curves are for $\Eb = 30$ dB and $\Eb = 20$ dB, both
of which correspond to the interference-limited regime where
$b^*=2.3$ bps/Hz ($N^* = 23$).
% and $\frac{\lambda_{\epsilon}^T (b^*)}{k} = 1.16$
The curves are nearly indistinguishable near the optimal $N^*$
because performance is essentially independent of $\Eb$ in the
interference-limited regime.
%The difference for very large values of
%$N$ is due to the fact that $\Ce \Ebp$ does depend on $\Ebp$,
%although these large values of $N$ are not important for optimized
%systems.
The middle set of curves correspond to $\rho= \Eb= 5$ dB, which is
between the two extremes. At this point $b^*=1.5$ bps/Hz ($N^* =
15$) and $\frac{\lambda_{\epsilon}^T (b^*)}{k} = 0.8$; reducing
power by 15 dB while keeping all other parameters fixed reduces
density/ASE by approximately a third.  The bottom curves correspond
to $\rho= \Eb = 0$ dB, which is in the wideband regime. At this
point $\Ce \Ebp=1$ and $b^*=(1-\frac{\alpha}{2})\Ce \Ebp=0.5$ bps/Hz
($N^*=5$), and the area spectral efficiency is reduced to $0.1$.

\begin{figure}
\centering
\includegraphics[width=4in]{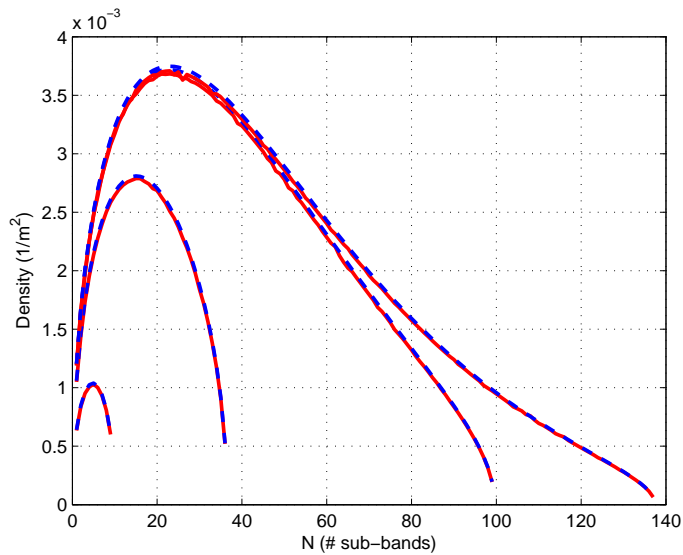}
\caption{Numerically computed $\lambda_{\epsilon}^T (N)$ versus $N$
for $\alpha=4$ and $\Eb=$ $30,~20,~5$, and $0$ dB (top to bottom).
Solid and dotted lines were computed using (\ref{eq-tc_approx1}) and
pure Monte Carlo simulation, respectively. } \label{fig-numerical}
\end{figure}

While $\Eb$ is generally thought to be adjusted by varying transmit
power, it can also be adjusted by fixing the power and varying the
rate $R$. The area spectral efficiency $\lambda_{\epsilon}^T (b^*)
\frac{R}{W}$ bps/Hz/${\textrm m}^2$ is equal to
$\left(\frac{-\log_e(1-\epsilon)}{\pi d^2} \right) b^*
\left(2^{b^*}-1 \right)^{- \frac{2}{\alpha}}$, and therefore depends
only on $\Eb$ but not on the particular values of $R$ and $\rho$. As
a result, if a network is operating outside of the
interference-limited regime, ASE can be tremendously increased by
either increasing power or decreasing rate (while keeping power
fixed).  However, this is only possible until the
interference-limited regime is reached; once there, ASE is
unaffected by $\Eb$.

%%%%%%%%%%%%%%%%%%%%%%%%%%%%%%%%%%%%%%%%%%%%%%%%%%%%%%%%%%%%%%%
\subsection{Direct Sequence Spread-Spectrum}

Direct-sequence (DS) spread spectrum is a well-established method
for spectrum sharing in wireless networks.  If DS is used with a
spreading factor of $N$, a signal with an information bandwidth
(i.e., symbol rate) of $\frac{W}{N}$ Hz can be spread across the
entire system bandwidth of $W$ Hz.  This is quite different than the
method investigated so far, which be thought of as either FDMA or
slow frequency-hopping.  In \cite{WebYan05} it is shown that DS is
significantly inferior to splitting the frequency band (FDMA) for
any particular bandwidth partition because it is preferable to avoid
interference (FDMA) rather than to suppress it (DS), and this
conclusion also holds if the bandwidth is optimally partitioned.

If direct sequence is used with completely separate despreading and
decoding (assuming spreading suppresses interference by a factor of
$N$), the SINR after despreading is given by:
\begin{eqnarray*}
\SINR_{\textrm{DS}} = \frac{ \rho d^{-\alpha}}{ N_0 (W/N) +
\frac{1}{N} \sum_{i \in \Pi(\lambda)} \rho X_i^{-\alpha}}.
\end{eqnarray*}
With some simple manipulation the outage probability is given by:
\begin{eqnarray}
\mathbb{P}[\SINR_{\textrm{DS}} \leq \beta(N)]
%&=&\mathbb{P}\left[\frac{ \rho d^{-\alpha}}{ N_0 (W/N) + \frac{1}{N}
%\sum_{i \in
%\Pi(\lambda)} \rho X_i^{-\alpha}} \leq \beta(N) \right] \\
 &=& \mathbb{P} \left[\frac{ \rho d^{-\alpha}}{ N_0 W  + \sum_{i \in
\Pi(\lambda)} \rho X_i^{-\alpha}} \leq \frac{\beta(N)}{N} \right]
\end{eqnarray}
where $\beta(N)$ is defined in (\ref{eq-beta}).  Therefore, the
total transmission density for DS with spreading factor $N$ is equal
to $\lambda_{\epsilon}$ as defined in (\ref{eq-transcap}) with
threshold $\frac{\beta(N)}{N}$ and $\eta = N_0 W$. However,
$\frac{\beta(N)}{N}$ is an \textit{increasing} function of $N$ and
thus total density monotonically \emph{decreases} with $N$ if DS is
used. Direct-sequence increases SINR by at most a factor of $N$, but
this gain is offset by the fact that the SINR threshold increases at
least linearly with $N$. As a result a DS system performs no better
than an FDMA/FH system with $N=1$, which corresponds to
$\lambda_{\epsilon}^T(1)$ in (\ref{eq-tc_approx1}) and which is
generally much smaller than the optimal $\lambda_{\epsilon}^T(N^*)$.
Although DS has strengths unrelated to spectral efficiency, such as
security and MAC design \cite{AndWebHae07}, these benefits come at a
significant performance penalty.

%%%%%%%%%%%%%%%%%%%%%%%%%%%%%%%%%%%%%%%%%%%%%%%%%%%%%%%%%%%%%%%%%%%%
\subsection{Effect of Frequency-Flat Fading and Variable TX-RX Distances} \label{sec-fading}

In the presence of fading and variable distances, the SINR
expression becomes:
\begin{eqnarray*}
\SINR_0 = \frac{ \rho d^{-\alpha} h_0} { \eta + \sum_{i \in
\Pi(\lambda)} \rho X_i^{-\alpha} h_i},
\end{eqnarray*}
where $h_i$ denotes the power of the fading coefficient from TX $i$
to the reference receiver, $h_0, h_1, \ldots$ are chosen iid
according to some distribution $F_H$, and $d$ is a random variable
chosen according to distribution $F_D$.  If we define 
%$Z$ as the normalized interference incorporating fading, i.e., 
$Z \triangleq \left( \sum_{i \in \Pi(1/\pi)} h_i Z_i^{-\alpha}
\right)^{-\frac{2}{\alpha}}$, and $G = d^{-\alpha}h_0$, then simple
manipulation yields:
\begin{eqnarray}
P_{\textrm{out}}(\lambda, \beta, \eta)
%&\triangleq&
% \mathbb{P} \left( \frac{ \rho d^{-\alpha} h_0}{ \eta + \rho \sum_{i \in
%\Pi(\lambda)}  h_i X_i^{-\alpha}} \leq \beta \right) \\
&=& \mathbb{P} \left( Z^{-\frac{\alpha}{2}} \geq \left( \pi \lambda
\right)^{-\frac{\alpha}{2}} \left( \frac{ G}{ \beta } - \frac{
\eta }{\rho}  \right) \right) \\
&=&  \mathbb{P} \left( G \leq \frac{ \beta \eta }{\rho} \right) +
\mathbb{P} \left( Z \leq \pi \lambda \left( \frac{ G}{ \beta } -
\frac{ \eta }{\rho} \right)^{-\frac{2}{\alpha}} \left| G \geq \frac{
\beta \eta }{\rho} \right. \right) \mathbb{P} \left( G \geq \frac{
\beta \eta }{\rho} \right) \label{eq-outage_fading_gen}
% &=&  \mathbb{P} \left( G \leq
%\frac{ \beta \eta }{\rho}\right) + \int_{\frac{ \beta \eta
%}{\rho}}^{\infty}  F_Z \left( \pi \lambda x^{-\frac{2}{\alpha}}
%\left( \frac{1 }{ \beta } - \frac{ \eta }{\rho x}
%\right)^{-\frac{2}{\alpha}} \right)  f_{G} (x) dx.
\end{eqnarray}
The first term is the probability of an outage due to insufficient
received signal power, i.e., $G$ is so small that the
interference-free SNR is below the SINR threshold, while the second
is the probability of outage conditioned on sufficient signal power.
Because of the somewhat involved expression for outage probability,
it is more convenient to consider bandwidth partitioning in terms of
outage minimization rather than density maximization.
%Note that the same definitions for the SINR threshold ($\beta(N) =
%2^{\frac{N R}{W}} - 1$) and noise power ($\eta = \frac{N_0 W}{N}$)
%apply here.
In the purely interference-limited regime ($N_0 = 0$), the first
term in (\ref{eq-outage_fading_gen}) disappears and the outage
probability (in terms of $N$) is given by:
\begin{eqnarray*}
P_{\textrm{out}}(N) &=& \mathbb{P} \left( Z G^{\frac{2}{\alpha}}
\leq \pi \left( \frac{\lambda}{N} \right) \left( 2^{\frac{NR}{W}} -
1 \right)^{\frac{2}{\alpha}} \right)  = \mathbb{P} \left( Z
G^{\frac{2}{\alpha}} \leq \pi \lambda \left( \frac{R}{W} \right)
\frac{1}{b} \left( 2^{b} - 1 \right)^{\frac{2}{\alpha}} \right),
\end{eqnarray*}
where we have again used $b = \frac{NR}{W}$. Outage is minimized by
minimizing $\frac{1}{b} \left( 2^{b} - 1
\right)^{\frac{2}{\alpha}}$, which is clearly equivalent to the
problem solved in Section \ref{sec-int_limited}.  Thus, \textit{the
interference-limited solution given in Theorem \ref{thm-int-limited}
is also optimal in the presence of fading and variable distances.}

However, the same is not true when there is positive noise power. By
substituting the appropriate values into
(\ref{eq-outage_fading_gen}) and manipulating the second addend,
outage is characterized as:
\begin{eqnarray} \nonumber
P_{\textrm{out}}(N) &=& \mathbb{P} \left( G \leq g^* \right) +
\int_{g^*}^{\infty}  F_Z \left( \pi \frac{\lambda}{N}
x^{-\frac{2}{\alpha}} \left( \frac{1 }{ 2^{\frac{NR}{W}} - 1 } -
\frac{ N_0 W }{N \rho x}
\right)^{-\frac{2}{\alpha}} \right)  f_{G} (x) dx \\
&=& \mathbb{P} \left( G \leq g^* \right) + \int_{g^*}^{\infty} F_Z
\left( \pi \lambda \frac{R}{W} x^{-\frac{2}{\alpha}} \frac{1}{b}
\left( \frac{1 }{ 2^b - 1 } - \frac{ N_0 R }{b \rho x}
\right)^{-\frac{2}{\alpha}} \right) f_{G} (x) dx
\label{eq-outage_fading}
\end{eqnarray}
where $g^* = \left(2^{\frac{NR}{W}} - 1 \right) \left( \frac{N_0
W}{N \rho} \right) = \left( \frac{2^b - 1 }{b} \right) \left(
\frac{N_0 R}{\rho} \right)$.  The first term, which represents
outage due to insufficient received power, increases with $N$
because $g^*$ is an increasing function of $N$.  The integrand in
the second term is the outage probability conditioned on $G = x$,
and is precisely of the form investigated earlier with $\Eb =
\frac{\rho x}{N_0 R}$. Therefore, Theorem \ref{thm-main}
characterizes the value of $N$ that minimizes the integrand for
\textit{each} value of $x$, but does not generally characterize the
minimizer of (\ref{eq-outage_fading}).
%not the integral itself; Theorem \ref{thm-main} also does not take
%into account received power-induced outages (i.e., the first term in
%(\ref{eq-outage_fading})).
However, the solution from Theorem \ref{thm-main} does become
increasingly accurate as transmission power is increased (i.e., the
interference-limited regime is approached) and as the variation in
the fading and TX-RX distances decreases. Increasing power causes
the first term in (\ref{eq-outage_fading}) to decrease and
eventually become negligible, while decreasing variation in $G$
reduces variation in the effective energy per bit $\Eb = \frac{\rho
x}{N_0 R}$.
% causes the integral in (\ref{eq-outage_fading}) to approach
%a value that effectively depends only on the average value of $G$.

To illustrate this, Fig. \ref{fig-fading} displays the outage
minimizing value of $N$ (computed via Monte Carlo) versus $\Eb$ for
four different settings: Rayleigh fading and Nakagami fading ($m=5$)
for fixed $d=10$, and no fading and variable distances for $d$
uniform in $[8,12]$ and $[5,15]$. The relevant parameters are:
$W=5$ MHz, $N_0 = 10^{-6}$, $R=1$ Mbps, $\lambda = \frac{.01}{\pi}$
$\textrm{m}^{-2}$, $\alpha=4$. The jitter in the curves is due to 
simulation error.  For sufficiently large $\Eb$, the
optimal does indeed converge to the optimal value for a purely
interference-limited ($N_0 = 0$) network.  Furthermore, the
optimizing $N$ tends towards the Theorem \ref{thm-main} solution for
more benign fading (Nakagami) and for smaller distance 
variation.\footnote{Our recent work has shown that there can be a 
substantial benefit to reducing variation in received signal power 
by adjusting transmit power to \textit{partially} compensate for 
reduced signal power \cite{FPC_Wireless}; thus, systems with relatively 
small signal power variation are particularly relevant.}
Based on (\ref{eq-outage_fading}), the truly optimal $N$ seems to
depend on the particular fading and distance distributions and appears
somewhat intractable; further investigation is left
for future work.

\begin{figure}
\centering
\includegraphics[width=4in]{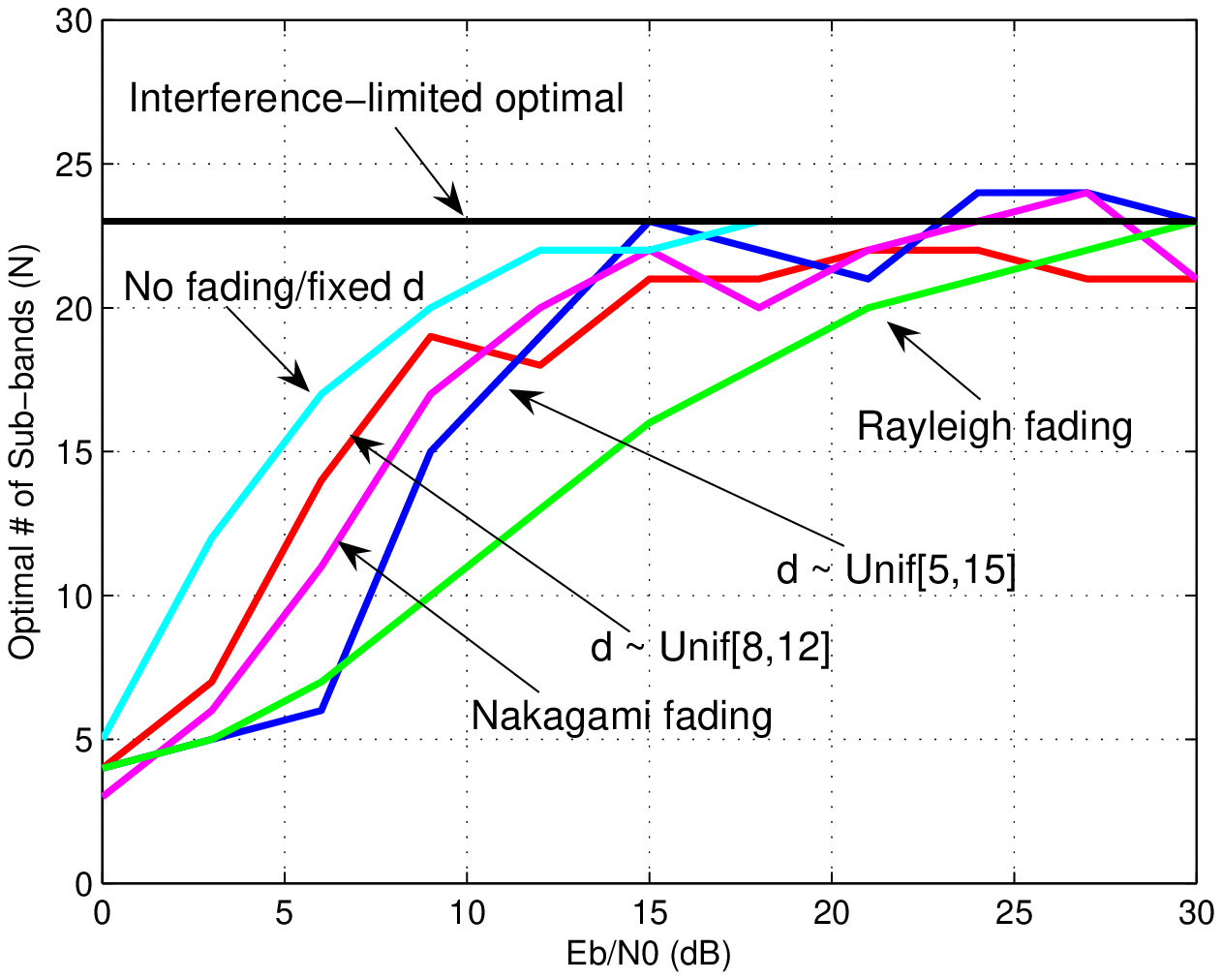}
\caption{Optimal Value of $N$ vs. $\Eb$} \label{fig-fading}
\end{figure}

%\footnote{The results of Theorem \ref{thm-main} do not
%hold only because of the variation in received \textit{signal}
%power, and not due to fading-induced variation in interference. If
%channel inversion is performed to completely compensate for fading
%and variable distance and equalize received signal power as
%investigated in \cite{WebAndJin07}, Theorem \ref{thm-main} applies
%for any $N_0 > 0$ because $G$ is deterministic.}

%%%%%%%%%%%%%%%%%%%%%%%%%%%%%%%%%%%%%%%%%%%%%%%%%%%%%%%%%%%%%%%%%%%%%
\section{Conclusion}

In this work we studied the problem of bandwidth partitioning
in a decentralized network and derived the optimal operating
spectral efficiency, assuming multi-user interference is
treated as noise and no transmission scheduling is performed.
A network can operate at this optimal point by dividing the total
bandwidth into sub-bands sized such that each transmission
occurs on one of the sub-bands at precisely the optimal spectral
efficiency.
%As a result, the optimal
%number of sub-bands is simply the optimal spectral efficiency (which
%is a deterministic function of the path loss exponent) divided by
%the normalized (by total bandwidth) rate.

The essence of this problem is determining the optimum balance between
the \emph{time-frequency resources} and \emph{area resources} consumed by
each transmission.  Using many time-frequency resources to transmit
a finite number of bits corresponds to operating at a low spectral efficiency.
This translates to a small required SINR, and thus only a small
area must be free of interfering transmissions. On the other hand, using few
time-frequency resources corresponds to a large spectral efficiency
and in turn a large SINR and interferer-free area.
Our analysis shows that the optimal depends only on the path loss exponent and
energy per information bit.  If thermal noise is negligible the
optimal spectral efficiency lies between the low- and high-SNR extremes,
while in the power-limited regime the optimal is a fraction of the
maximum possible spectral efficiency in the absence of interference.
Furthermore, the optimal spectral efficiency is always an increasing
function of the path loss exponent and of the energy per information bit.

%%%%%%%%%%%%%%%%%%%%%%%%%%%%%%%%%%%%%%%%%%%%%%%%%%%%%%
\appendices

%%%%%%%%%%%%%%%%%%%%%%%%%%%%%%%%%%%%%%%%%%%%%%%%%%%%%%%%
\section{Proof of Theorem \ref{thm-main}} \label{app-proof-main}

In order to prove the result it is convenient to work with natural
logarithms:
\begin{eqnarray}
\lambda_{\epsilon}^T (b) = \kappa b \left( \frac{1}{2^b-1} -
\frac{1}{b} \frac{1}{\Eb} \right)^{\frac{2}{\alpha}}
%&=& \frac{\kappa}{\log_e 2} b \log_e 2
% \left( \frac{1}{e^{\log_e 2 b}-1} - \frac{1}{b \log_e 2} \frac{1}{\Eb/\log_e 2} \right)^{\frac{2}{\alpha}} \\
&=& \left(\frac{\kappa}{\log_e 2}\right) n  \left( \frac{1}{\ex^n-1}
- \frac{1}{n} \frac{1}{\Ebt} \right)^{\frac{2}{\alpha}} =
\lambda_{\epsilon}^T (n)
\end{eqnarray}
where $n \triangleq b \log_e 2$ and $\Ebt \triangleq
\frac{\Eb}{\log_e 2}$.
%We can characterize the optimizing value
%$n^*$, and then use $b^* = \frac{n^*}{\log_e 2}$.
Ignoring constant $\frac{\kappa}{\log_e 2}$ and defining $\delta =
\frac{2}{\alpha}$, the first derivative is:
%of $\lambda_{\epsilon}^T(n)$ is:
\begin{eqnarray}
\frac{d}{dn}\left[ \lambda_{\epsilon}^T(n) \right] &=&
\left(\frac{1}{\ex^n-1}-\frac{1}{\Ebt n} \right)^{\delta} + n \delta
\left(\frac{1}{\ex^n-1}-\frac{1}{\Ebt n} \right)^{\delta-1}
\left(\frac{-\ex^n }{\left(\ex^n-1 \right)^2}+\frac{1}{\Ebt n^2}
\right) \\
&=& \frac{ \left(\frac{1}{\ex^n-1}-\frac{1}{\Ebt n} \right)^{\delta
- 1} }{\Ebt n \left(\ex^n-1 \right)^2} \left[ \Ebt n \left(\ex^n-1
\right)  - \Ebt \delta n^2  \ex^n - \left(1-\delta \right)
\left(\ex^n-1 \right)^2 \right].
\end{eqnarray}
Because the first term is positive for any $n>0$, the derivative is
equal to zero if and only if:
\begin{equation}  \label{eq-fixed_n}
\Ebt n \left(\ex^n-1 \right)  - \Ebt \delta n^2  \ex^n -
\left(1-\delta \right) \left(\ex^n-1 \right)^2 = 0.
\end{equation}
Substituting $n = b \log_e 2$ and $\Ebt = \frac{\Eb}{\log_e 2}$
yields the fixed point equation in (\ref{eq-main}). Although
$\lambda_{\epsilon}^T(n)$ is neither convex nor concave, we can show
it has a unique maximizer at the unique positive solution to the
above equation.  It is easy to check that $\lambda_{\epsilon}^T(0) =
\lambda_{\epsilon}^T(n_{\rm max}) = 0$ and $\lambda_{\epsilon}^T(n)
> 0 $ for $0 < n < n_{\rm max}$, where $n_{\rm max}= \Ce \Ebp \log_e
2$.
% and therefore satisfies $\frac{\ex^{n_{\rm max}}-1}{n_{\rm max}} = \Ebt$.
Therefore the function is maximized at a point where its derivative
is zero. Furthermore, (\ref{eq-fixed_n}) is satisfied at any point
where the derivative is zero and thus (\ref{eq-fixed_n}) must have
at least one positive solution. To show that (\ref{eq-fixed_n}) has
a unique positive solution, define
%\begin{equation}
$\nu(n) = \Ebt \delta n^2 \ex^n + (1-\delta) \left(\ex^n-1
\right)^2$ and $\nu(n) = \Ebt n \left(\ex^n-1 \right)$.
%\end{equation}
Equation (\ref{eq-fixed_n}) is satisfied if and only if $\mu(n) =
\nu(n)$. Note that $\mu(0) = \nu(0) = 0$ and $\mu(n) > 0$ and
$\nu(n) > 0$ for all $n > 0$. Simple calculations show that each
function is strictly convex. Hence $\mu(n),\nu(n)$ are positive
valued, non-decreasing, strictly convex functions, and based on this
it is straightforward to argue that there is at most one positive
solution of $\mu(n) = \nu(n)$.

To show $n^*(\Ebt, \delta)$ increases with $\Ebt$, define the LHS of
(\ref{eq-fixed_n}) as $f(n,\Ebt,\delta)$. By the properties shown
earlier, $f(n,\Ebt,\delta)>0$ for $0 < n < n^*(\Ebt,\delta)$ and
$f(n,\Ebt,\delta) < 0$ for $n > n^*(\Ebt,\delta)$.  As a result,
$n^*(\Ebt, \alpha)$ increases with $\Ebt$ if for any $\Ebt'
> \Ebt$, $f(n,\Ebt',\delta)>0$ for $0 < n < n^*(\Ebt,\delta)$.  To
prove this property, choose any $n, \Ebt$ such that
$f(n,\Ebt,\delta) = \Ebt \left( n \left(\ex^n-1 \right) - \delta n^2
\ex^n \right) - \left(1-\delta \right) \left(\ex^n-1 \right)^2 > 0$.
Since $\left(1-\delta \right) \left(\ex^n-1 \right)^2 > 0$ for any
$n$, this implies $n \left(\ex^n-1 \right) - \delta n^2 \ex^n > 0$.
Thus, for any $\Ebt' > \Ebt$:
\begin{eqnarray}
f(n,\Ebt',\delta) &=& \Ebt' \left( n \left(\ex^n-1 \right) - \delta
n^2 \ex^n \right) - \left(1-\delta \right) \left(\ex^n-1 \right)^2
> f(n,\Ebt,\delta) > 0.
%\Ebt \left( n \left(\ex^n-1 \right) - \delta n^2 \ex^n \right) -
%\left(1-\delta \right) \left(\ex^n-1 \right)^2 > 0.
\end{eqnarray}

By a similar argument, if $f(n,\Ebt,\delta)$ is a
\textit{decreasing} function of $\delta$ then $n^*(\Ebt, \delta)$
decreases with $\delta$, i.e., increases with $\alpha$.  To prove
this, note that the partial of $f(n,\Ebt,\delta)$ with respect to
$\delta$ is $\left(\ex^n-1 \right)^2 - \Ebt n^2 \ex^n$.  
Recall that $n \leq n^{\rm max}$ is equivalent to $\Ebt n \geq \ex^n-1$.  This allows:
\begin{equation}
\Ebt n^2 \ex^n - \left(\ex^n-1 \right)^2 \geq \left(\ex^n-1 \right)n \ex^n - \left(\ex^n-1 \right)^2 = \left(\ex^n-1 \right)\left(n \ex^n - \ex^n+1 \right) \geq 0.
\end{equation}
The last expression is nonnegative on account of the fact that the function 
$n \ex^n - (\ex^n-1)$ has derivative $n \ex^n \geq 0$.  
Thus $f(n,\Ebt,\delta)$ is decreasing in $\delta$, i.e., increasing in $\alpha$.

\section{Proof of Theorem \ref{thm-wideband}}
\label{app-wideband}

For convenience, we again work with the function in natural log form
(see Appendix \ref{app-proof-main}).  To prove the result, we expand the
exponential terms (using $\ex^x = \sum_{k=0}^{\infty} \frac{x^k}{k!}$) in
% the natural-log version of the fixed point equation
(\ref{eq-fixed_n}) to give:
\begin{equation}
\Ebt \delta n^2 (1 + n + O(n^2)) -
\Ebt n \left(n + \frac{n^2}{2} + O(n^3) \right) +
(1-\delta) \left( n^2 + n^3 + O(n^4)\right)^2 = 0.
\end{equation}
Cancelling a factor of $n^2$ throughout yields
\begin{eqnarray}
\Ebt \delta \left( 1 + n + O(n^2) \right) -
\Ebt \left(1 + \frac{1}{2}n + O(n^2) \right) +
(1-\delta)\left(1 + n + O(n^2) \right) = 0,
\end{eqnarray}
which can be solved to yield a solution that is accurate to within a quadratic term:
\begin{eqnarray}
n^* &=& \frac{\Ebt(1-\delta) + (\delta-1)}
{\Ebt \left(\delta-\frac{1}{2}\right)+(1-\delta)} + O(n^2).
\end{eqnarray}

We are interested in the behavior of $\frac{b^*}{\Ce \Ebp}$ as $\Eb
\rightarrow 0$ (or equivalently $\Ce \Ebp \rightarrow 0$).  Because
$n^* = b^* \log_e 2$ and $n_{\rm max} = \Ce \Ebp \log_e 2$ we can
equivalently evaluate $\frac{n^*}{n_{\rm max}}$
\begin{eqnarray}
\lim_{n_{\rm max} \rightarrow 0}  ~ \left( \frac{n^*}{n_{\rm max}} \right)
= \lim_{n_{\rm max} \rightarrow 0}  ~
\frac{1}{n_{\rm max}} \left(
\frac{\Ebt(1-\delta) + (\delta-1)}
{\Ebt \left(\delta-\frac{1}{2}\right)+(1-\delta)} \right).
\end{eqnarray}
By plugging in $\frac{\ex^{n_{\rm max}}-1}{n_{\rm max}} = \Ebt$  and
using L'Hospital's rule, we can show the above limit is $1-\delta$,
which implies $n^* = (1 - \delta) n_{\rm max} + O(n^2)$, which in turn
gives the final result:
\begin{eqnarray*}
b^* = \left(1 - \frac{2}{\alpha} \right) \Ce \Ebp\left( \Eb \right)
+ O \left( b^2 \right).
\end{eqnarray*}

Because our approximation is accurate within a quadratic, we have the following:
\begin{eqnarray}
\lim_{\Ce \Ebp \rightarrow 0}  ~
\frac{\lambda_{\epsilon}^T(b^*)}{\Ce \Ebp} = \lim_{\Ce \Ebp
\rightarrow 0}  ~ \frac{\lambda_{\epsilon}^T\left( \left(1 - \delta
\right) \Ce \Ebp\right)}{\Ce \Ebp}.
\end{eqnarray}
By working with the natural log version of this equation and
plugging in $\frac{\ex^{n_{\rm max}}-1}{n_{\rm max}} = \Ebt$,
L'Hospital's rule can be used to show that this limit is equal to
$\kappa (1-\delta)^{(1-\delta)} \delta^{\delta} 2^{-\delta}$, which
yields (\ref{eq-lambda_wide}).

\end{document}